\documentclass[preprint,12pt]{elsarticle}

\usepackage{amsmath,amsfonts,amsthm,amssymb,paralist,subfigure,graphicx,amsbsy,float,epsfig,color}
\usepackage{cuted,mathtools,lipsum}

\usepackage[ruled,vlined,linesnumbered]{algorithm2e}

\usepackage{stmaryrd,url}

\usepackage{amsthm}
\newtheorem{lem}{Lemma}
\newtheorem{ass}{Assumption}
\newtheorem{theorem}{Theorem}

\newtheorem{rem}{Remark}

\def\mb{\mathbf}

\def\mc{\mathcal}

\def\mb{\mathbf}

\def\mc{\mathcal}

\DeclareMathOperator*{\argmin}{argmin}
\DeclareMathOperator*{\argmax}{argmax}

\journal{Systems and Control Letters}

\begin{document}

\begin{frontmatter}

\title{Momentum-based Distributed Resource Scheduling Optimization Subject to Sector-Bound Nonlinearity and Latency
}

\author[Sem]{Mohammadreza Doostmohammadian}
\affiliation[Sem]{Mechatronics Group, Faculty of Mechanical Engineering, Semnan University, Semnan, Iran, and Center for International Scientific Studies and Collaborations, Tehran, Iran, doost@semnan.ac.ir.}

\author[kaz]{ Zulfiya R. Gabidullina}
\affiliation[kaz]{Institute of Computational Mathematics and Information Technologies, Kazan Federal University, Russia, Zulfiya.Gabidullina@kpfu.ru.}

\author[HR]{ Hamid R. Rabiee}
\affiliation[HR]{Computer Engineering Department, Sharif University of Technology, Tehran, Iran,
	rabiee@sharif.edu.}

\begin{abstract}
This paper proposes an accelerated consensus-based distributed iterative algorithm for resource allocation and scheduling.  The proposed gradient-tracking algorithm introduces an auxiliary variable to add momentum towards the optimal state. We prove that this solution is all-time feasible, implying that the coupling constraint always holds along the algorithm iterative procedure; therefore, the algorithm can be terminated at any time. This is in contrast to the ADMM-based solutions that meet constraint feasibility asymptotically. Further, we show that the proposed algorithm can handle possible link nonlinearity due to logarithmically-quantized data transmission (or any sign-preserving odd sector-bound nonlinear mapping).
We prove convergence over uniformly-connected dynamic networks (i.e., a hybrid setup) that may occur in mobile and time-varying multi-agent networks. Further, the latency issue over the network is addressed by proposing delay-tolerant solutions. To our best knowledge, accelerated momentum-based convergence, nonlinear linking, all-time feasibility, uniform network connectivity, and handling (possible) time delays are not \textit{altogether} addressed in the literature. These contributions make our solution practical in many real-world applications.
\end{abstract}

\begin{graphicalabstract}
	\includegraphics{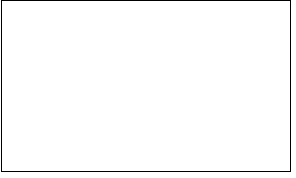}
\end{graphicalabstract}

\begin{highlights}
	\item Introduction of a novel accelerated momentum-based distributed iterative algorithm specifically designed for optimal resource allocation and scheduling tasks.
	\item Demonstration that the proposed algorithm maintains all-time feasibility, ensuring that the coupling constraint is satisfied at every iteration, allowing for termination at any point without violating the constraints.
	\item Capability to manage link nonlinearity arising from logarithmically-quantized data transmission or other sign-preserving odd sector-bound nonlinear mappings.
	\item  Proof of convergence over uniformly-connected dynamic networks, accommodating the complexities of mobile and time-varying multi-agent systems.
	\item Addressing latency issues in network communication by proposing solutions that are tolerant to delays, enhancing practical applicability.
\end{highlights}

\begin{keyword}
			Distributed constrained optimization\sep consensus\sep graph theory\sep convex analysis\sep resource scheduling and allocation
\end{keyword}

\end{frontmatter}

\section{Introduction} \label{sec_intro}
Decentralized multi-agent optimization and learning has been significantly advanced due to recent developments in cloud computing, parallel processing, Internet-of-Things (IoT) and applications in machine learning, artificial intelligence, and distributed data processing \cite{de2019survey,heydaribeni2019distributed}. This motivates decentralized setups toward optimizing the resource scheduling on large-scale, where the process of learning is distributed over a network of computing nodes/agents \cite{2025survey} (see Fig.~\ref{fig_distributed}).
The challenges posed by the dynamic workloads, data-transmission latencies, volatile networking conditions, and nonlinearities (such as quantization and clipping) within the system ask for innovative solutions that can rapidly adapt to the intricacies of modern computing problems. This motivates to bridge the gap and devise a solution that can operate effectively and rapidly in the face of different local and coupling constraints and ensure the reliability of distributed applications over networked setups.
In response to the mentioned challenges, this study introduces a momentum-based algorithm that not only addresses the intricacies of existing constraints but also provides a dynamic localized framework for resource allocation and scheduling. The motivation is to offer a distributed strategy that collaboratively optimizes resource utilization in the presence of latency, and effectively goes through the nonlinearities inherent in real-world computing setups while addressing fast convergence.
\begin{figure}[]
\centering
\includegraphics[width=3.3in]{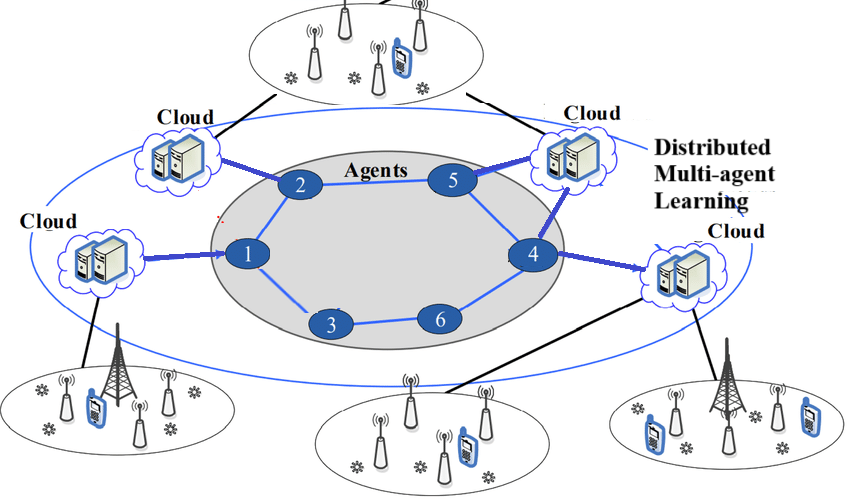}
\caption{A distributed networked multi-agent optimization/learning setup motivated by cloud-based decentralized computing and parallel processing.  }
\label{fig_distributed}
\end{figure}

Some relevant literature is discussed here. The optimal resource-constrained problem in centralized setups using transfer learning is considered in \cite{askarizadeh2023resource}. The preliminary works on distributed linear resource scheduling algorithms are discussed in \cite{boyd2006optimal,shames2011accelerated}. Distributed optimization application by combining finite-time observers and homogeneous system theory with control design application is considered in \cite{wang2022distributed}. Distributed nonlinear dynamics for optimization using single-bit consensus protocols (to reduce communication load on agents) is proposed in \cite{taes}.
Optimal energy scheduling for plug-in electric vehicle (PEV) charging is discussed in \cite{falsone2020tracking,vujanic2016decomposition}. Energy resource allocation over the power grid (referred to as the economic dispatch problem) is considered using different consensus-based \cite{chen2016distributed,cherukuri2015distributed,liu2021distributed}, event-driven \cite{jena2023securing}, and adaptive dynamic programming methods \cite{venayagamoorthy2016dynamic}. Distributed multi-robot task assignment (and allocation) is considered in  \cite{su2018task,zhang2023energy,MiadCons}.  Momentum-based solution for \textit{unconstrained} optimization with machine learning applications is studied in \cite{tase,carnevale2022gtadam,xin2019distributed,okamoto2023distributed}.
Distributed consensus coordination and optimization under time-delay is studied in \cite{jin2023new,liu2022distributed,wang2022distributed2}.
However, these works mostly require all-time network connectivity (i.e., cannot handle temporary disconnectivity), and further cannot address all-time resource-demand feasibility while unable to handle nonlinear data-transmission (e.g., due to log-quantization). Quantized communication over all-time connected networks (with no all-time constraint feasibility) is studied in \cite{RIKOS20235900}. All-time feasible methods are proposed in \cite{wu2021new,cherukuri2015distributed} without consideration of communication delays and nonlinearity.
Alternating-direction-method-of-multipliers (ADMM) resource allocation solutions \cite{carnevale2022gtadam,jian2019distributed,banjac2019decentralized,jiang2022distributed,cdc_dtac} are based on primal-dual optimization. These works are known to reach constraint-feasibility asymptotically over time and therefore are not all-time feasible. Further, they only converge on all-time connected networks and are mostly vulnerable to time-delays.  Another interesting line of research is additive-	increase-multiplicative-decrease (AIMD) strategies used in various networking and resource management scenarios, particularly for controlling congestion in data networks \cite{AIMD_book}; the literature includes average-based AIMD admission control policy with minimal communication between individual nodes and the aggregator \cite{AIMD2} and its application in charging electric and plug-in hybrid vehicles \cite{Stdli2012AIMDlikeAF} and optimal scheduling of the charging process of electric bus fleets \cite{AIMD1}. Finally, this paper improves the previous work of the authors \cite{scl} by adding a momentum term to reach faster convergence.

This paper introduces a novel momentum-based algorithm designed to tackle the mentioned challenges associated with distributed resource allocation and scheduling, offering an adaptive solution to the dynamic nature of the existing applications. The main contributions are summarized as follows:
\begin{enumerate}[(i)]
\item Our proposed algorithm operates over uniformly-connected networks, where the agents may lose connectivity temporarily because of mobility or packet drops. This exhibits a remarkable ability to converge under time-varying and relaxed network-connectivity conditions. This advances the existing ADMM-based \cite{carnevale2022gtadam,jian2019distributed,banjac2019decentralized,jiang2022distributed,cdc_dtac} and delay-tolerant solutions \cite{wang2022distributed2,liu2022distributed}.
\item The algorithm holds exact convergence (with no optimality gap) in the face of possible nonlinearities commonly encountered in practical distributed systems, such as logarithmic quantization and saturation/clipping. In general, the proposed solution is robust subject to general sign-preserving sector-bound nonlinear mapping on the data-transmission channels. This adaptability enhances its applicability, from traditional computing clusters to emerging edge and fog computing paradigms. This contribution is not addressed in most existing literature.
\item The proposed solution is all-time constraint feasible, implying that at any termination time of the algorithm, the resource-demand balance holds. This contributes to its ability to handle dynamic setups and avoid service disruption. This particularly is in contrast to the dual-based ADMM  \cite{carnevale2022gtadam,jian2019distributed,banjac2019decentralized,jiang2022distributed,cdc_dtac} and other existing delay-tolerant solutions \cite{wang2022distributed2,liu2022distributed}.
\item By leveraging momentum-based techniques, our approach enables rapid convergence to optimal resource distribution. This makes it particularly well-suited for real-time scenarios, ensuring that system resources are optimally allocated sufficiently fast. This along with the all-time feasibility, enhances the algorithm's adaptability towards scenarios characterized by evolving workloads.
\item Another key strength of our proposed algorithm lies in its capability to handle latency and time-delay prevalent in data-transmission networks. This is mainly due to limited bandwidth and synchronization issues in the data-sharing network of processing nodes. Our approach ensures optimality even in situations characterized by variable heterogeneous time-delays over undirected symmetric networks. This feature is particularly crucial in real-world applications where instantaneous decision-making is imperative. This notion of network latency is rarely addressed in distributed scheduling literature. Few existing works \cite{wang2022distributed2,liu2022distributed} are not all-time feasible and cannot handle nonlinearity.
\end{enumerate}

To our best knowledge, no existing work in the literature addresses the contributions (i)-(v) altogether. Some works address one or two contributions (as discussed in the literature review), but not all the mentioned contributions entirely together.

\textit{Paper Organization:} Section~\ref{sec_pre} presents some preliminaries and formulated the problem. Section~\ref{sec_alg} presents our main algorithm for iterative resource scheduling. Section~\ref{sec_conv} discusses the constraint feasibility and optimal convergence. Section~\ref{sec_sim} provides some simulations and Section~\ref{sec_con} concludes the paper.

\textit{General Notations:} $\mb{1}_n$ and $I_n$ respectively denote the vector of all ones and identity matrix of size $n$. The operator ``$;$'' implies column concatenation. $\partial_x$ is simplified form of $\frac{d}{dx}$, i.e., the derivative with respect to $x$. $\nabla_x F$ is the gradient of $F(\cdot)$ with respect to $x$.

\section{Preliminaries and Problem Formulation} \label{sec_pre}	
\subsection{Problem Setup}
The resource scheduling and allocation problem is formulated as equality-constraint optimization in its primal form. It aims to minimize a globally convex objective function subject to the resource-demand constraint. This balancing constraint guarantees that the weighted sum of resources meets the sum of demands, and is referred to as the \textit{feasibility constraint}. Violating this constraint may cause service disruption in practical applications. The mathematical form of the problem is
\begin{equation} \label{eq_dra0}
\min_{\mb{z}}
~~  \widetilde{F}(\mb{z}) := \sum_{i=1}^{n} \widetilde{f}_i(z_i), ~~
\text{s.t.} ~  \mb{a}^\top\mb{z} - b = 0,
\end{equation}
with the column state vector $\mb{z} =[z_1;\dots;z_n] \in \mathbb{R}^n$ representing the resources ($z_i \in \mathbb{R}$ is the resource at node/agent $i$),
and positive weighting factor $\mb{a}=[{a}_1;\dots;{a}_n] \in \mathbb{R}^n_{+}$ and the demand vector $b \in \mathbb{R}$. The function $f_i: \mathbb{R} \mapsto  \mathbb{R}$ denotes the local objective/cost at node $i$.
\begin{ass}
The global cost function $\widetilde{F}(\cdot)$ is strictly convex and smooth.
\end{ass}
By simple change of variables $ x_i :=z_i a_i$, the problem takes the standard mathematical form as,
\begin{equation} \label{eq_dra}
\begin{aligned}
	\displaystyle
	\min_{\mb{x}}
	~~ & F(\mb{x}) = \sum_{i=1}^{n} f_i(x_i) \\
	\text{s.t.} ~~&  \sum_{i=1}^{n} (x_i -b_i) =0
\end{aligned}
\end{equation}
with $\mb{x} =[x_1;\dots;x_n] \in \mathbb{R}^n$ as the new state variable and the predefined parameter $b_i \in \mathbb{R}$ as the preset local demand at node $i$ (note that $\sum_{i=1}^{n} b_i = b$). For example, in power allocation and economic dispatch, the parameter $b_i$ represents the nominal power rating of the power generators and is a predefined value.
\textit{The above equality constraint represents the balance between the scheduled resources and the demand, which needs to be always satisfied in all-time feasible solutions.}
In the case that the resource states are further locally constrained by $m_i \leq x_i \leq M_i$ and $m_i \leq b_i \leq M_i$, the local objective $f_i(\cdot)$ is added with a penalizing additive term to address these box constraints \cite{bertsekas1975necessary,wu2021new}.
One example of penalty function is $\sigma [{x}_i - M_i]^+ + \sigma [m_i - {x}_i]^+$
with $[u]^+ = \max \{u,0\}^c$, $c \in \mathbb{Z}^+$.
Another smooth example is,
\begin{align} \label{eq_log_penalty}
\frac{\sigma}{\alpha} \log (1+e^{\alpha ({x}_i - M_i)}) + \frac{\sigma}{\alpha} \log (1+e^{\alpha (m_i - {x}_i)}),
\end{align}
with $\sigma,\alpha \in \mathbb{R}^+$.

\begin{rem}
	In the problem statement \eqref{eq_dra} the state variables are considered scalar as in many applications mentioned in Section \ref{sec_app}, including the energy resource management (generator coordination) and CPU scheduling, where the optimization state variable is scalar, for example, the allocated power in generator coordination and CPU workload in CPU scheduling over data centres. This scalar case is also considered in many works in the literature, for example, see references \cite{boyd2006optimal,shames2011accelerated,chen2016distributed,cherukuri2015distributed,scl,rikos2021optimal,grammenos2023cpu,liu2023distributed}. However, the results can be easily extended to the case where $x_i \in \mathbb{R}^p$ and $b_i \in \mathbb{R}^p$ with $p>1$.
\end{rem}

\subsection{Applications} \label{sec_app}
The applications of the optimization problem~\eqref{eq_dra} range from the control systems to computer networks. In these applications a group of agents collaboratively learn how to locally optimize the resource scheduling costs, as discussed in the following:
\begin{itemize}
\item An example is the problem of distributed CPU scheduling over networked data centres considered in \cite{rikos2021optimal,kwak2015dream,grammenos2023cpu}. In this problem, the idea is to optimally allocate tasks to a group of networked servers while optimizing the computational cost functions. The equality constraint balances the assigned CPU resources and the computational tasks.
\item Distributed energy resource management and generator coordination is another application considered in \cite{chen2016distributed,cherukuri2015distributed,venayagamoorthy2016dynamic,liu2023distributed} to distribute the electricity dispatch and power demand to a group of networked power generators. The resource-demand feasibility implies the balance between the assigned power resources to the load demand.
\item Battery scheduling for distributed PEV charging is studied in \cite{falsone2020tracking,xie2016fair,vujanic2016decomposition} to assign a fleet of PEVs to charging stations to optimize the electricity cost. The decision-making is distributed and localized to address large-scale frameworks.
\item  The problem of rate-distributed linearly constrained minimum variance (LCMV) beamforming in wireless acoustic sensor networks is another example in \cite{zhang2019distributed,koutrouvelis2018low}. The goal is to optimize the energy usage within the network and
the transmission cost while addressing the constraint on noise reduction performance.
\item Another application is in coverage control, facility location optimization, and distributed task allocation over mobile robotic networks \cite{MiadCons,gao2022coverage,MSC09,jin2016distributed}. The idea is to optimally assign a group of agents to maximize the mission area coverage while balancing duty-to-capability ratios.
\end{itemize}

\subsection{Graph Theoretic Notions} \label{sec_graph}
The multi-agent network is modelled by a graph topology $\mc{G} = \{\mc{V},\mc{E}\}$ with adjacency (weight) matrix $W$. Its associated Laplacian matrix is defined as $L=D-W$ with $D = \mbox{diag}[\sum_{j=1}^n W_{ij}]$ (the diagonal degree matrix).
Assuming the network to be undirected and $W$ to be symmetric, its  $L$ is positive semi-definite (PSD) \cite{SensNets:Olfati04}. This is the case when the agents have similar communication range and broadcasting power.
The graph is said to be weight-balanced (WB) if $\sum_{i=1}^n W_{ij}= \sum_{j=1}^n W_{ij}$. For WB networks, the matrix $L_s:=\frac{L+L^\top}{2}$ is PSD.
It is possible that the agents temporarily lose connectivity due to mobility or packet drops. In this case, the switching network $\mc{G}(t)$ is said to be \textit{uniformly connected} over time-window $B>0$ (or $B$-connected), if there is $B>0$ such that $\mc{G}_B(t)$ with edge union $\bigcup_{t}^{t+B}\mc{E}(t) $ is connected. In this case, the Laplacian $L_B$ associated with $\mc{G}_B$ only has one isolated zero eigenvalue \cite{godsil}. Note that this is the most relaxed network connectivity condition and contradicting this condition implies that, for example, there are two separate subgraphs (islands) in the network with no path between them, and therefore no collaborative agreement can be reached over a non-uniformly-connected network. This union connectivity is particularly vital for link-failure scenarios. In case there are packet drops or link failures with probability $p_l$, one can define a \textit{bond-percolation} probability $p_c$ over which the network loses its connectivity. In other words, $p_c$ serves as a critical threshold that separates the
two network connectivity and disconnectivity phases. For example, for Erdos-Renyi (ER) random graphs with linking density $p$ and link failure rate $p_l$, for $p_l<p_c$ the outcome network is known to be islanded (consists of two or more non-connected subgraphs) \cite{karonski1995random}. For ER networks of size $n$ this bond-percolation threshold is \cite{kawamoto2015precise,kitsak2010identification}
\begin{align} \label{eq_bond}
p_c = 1-\frac{1}{\overline{d}_{\mc{G}}}
\end{align}
with $\overline{d}_{\mc{G}}=\frac{p(n-1)}{2}$ as the average node degree.

\begin{ass}
	The network $\mc{G}(t)$ is (potentially) time-varying, weight-balanced
	(WB), and uniformly-connected over time.
\end{ass}

\subsection{Preliminary Lemmas and Auxiliary Results}

\begin{lem} \label{lem_z*}
\cite[Lemma~2.3]{cherukuri2015distributed} The state $\mb{x}^*$ as the optimizer of the problem~\eqref{eq_dra} satisfies $\nabla_x F(\mb{x}^*) \in \mbox{span}(\mb{1}_n)$.
\end{lem}
The proof of the above lemma directly follows the Karush-Kuhn-Tucker (KKT) conditions on the affine constraint $\sum_{i=1}^{n} (x_i -b_i) =0$.

\begin{lem}  \label{lem_xLy}
Let  matrix $W$ be symmetric, ${\overline{\mb{x}} := \mb{x} - \frac{\mb{1}_n^\top \mb{x}}{n} \mb{1}_n}$, and $g_l: \mathbb{R} \mapsto \mathbb{R}$ be a sign-preserving odd nonlinear mapping with sector bounds $\kappa \leq \frac{g_l(x_i)}{x_i} \leq \mc{K}$ ($\kappa,\mc{K} \in \mathbb{R}^+$).
Then,
\begin{align}    \label{eq_laplace}
	\lambda_2 \|\overline{\mb{x}} \|_2^2 \leq &\mb{x}^\top L \mb{x} = \overline{\mb{x}}^\top L \overline{\mb{x}} \leq \lambda_n \|\overline{\mb{x}} \|_2^2 \\ \label{eq_laplace2}
	\lambda_2 \kappa \|\overline{\mb{x}} \|_2^2 &\leq g_l(\mb{x})^\top L \mb{x} \leq \lambda_n \mc{K} \|\overline{\mb{x}} \|_2^2
\end{align}
with $\lambda_n,\lambda_2$ respectively as the largest and smallest non-zero eigenvalues of the laplacian matrix $L$ (or $L_B$).
\end{lem}
\begin{proof}
The proof of \eqref{eq_laplace} is provided in  \cite{olfatisaberfaxmurray07}. To prove \eqref{eq_laplace2}, let ${\overline{g_l(\mb{x})} := g_l(\mb{x}) - \frac{\mb{1}_n^\top g_l(\mb{x})}{n} \mb{1}_n}$,
\begin{align} \nonumber
	g_l(\mb{x})^\top L \mb{x} &= \overline{g_l(\mb{x})}^\top L \overline{\mb{x}} \\ \label{eq_proof_Ls}
	&=  \frac{1}{2}\sum_{i,j=1}^n W_{ij}(g_l(x_i)-g_l(x_j))(x_i-x_j)
\end{align}
Recalling the sector-bound property of $g_l(\cdot)$,
\begin{align}\nonumber
	\kappa (x_i-x_j)(x_i-x_j) \leq  (g_l(x_i)&-g_l(x_j))(x_i-x_j) \\
	&\leq \mc{K}(x_i-x_j)(x_i-x_j)
\end{align}
Substituting the above equation in \eqref{eq_proof_Ls} and using  \eqref{eq_laplace}, the proof of \eqref{eq_laplace2} follows.
\end{proof}


\begin{lem} \label{lem_strict}
\cite{Boyd-CVXBook} For the strictly-convex function $F:\mathbb{R}^n \mapsto \mathbb{R}$ with Lipschitz gradient satisfying $0 <  \frac{d^2 f_i(x_i)}{dx_i^2} < 2 u$, define ${\gamma(k) := \mb{x}(k+1)-\mb{x}(k)}$ with $\mb{x}(k)$ and $\mb{x}(k+1)$ as the allocated node states at time step $k$ and $k+1$, respectively. We have,
\begin{align}
	F(\mb{x}(k+1)) \leq F(\mb{x}(k)) + \nabla_x F(\mb{x}(k))^\top \gamma(k) +  u\gamma(k)^\top \gamma(k)
	\label{eq_taylor_2}
\end{align}
\end{lem}

\section{The Proposed Iterative Algorithm} \label{sec_alg}
\subsection{Momentum-based Solution}
The state $x_i$ at every node/agent $i$ denotes the amount of resources assigned to that agent. In the distributed setup, agents cooperatively share necessary information on the gradients (a gradient-tracking solution) with their neighbour agents to iteratively optimize the cost function. The multi-agent network $\mc{G}$ is assumed undirected\footnote{As it is discussed later, in the absence of time-delays over the network, the decentralized solution works over general WB directed networks.} with symmetric weight matrix $W$.	
The proposed state update at every node $i$ is as follows:	

\small \begin{align}
x_i(k+1) &= x_i(k) +\eta  \sum_{j \in \mc{N}_i}   W_{ij}  (\partial_{x_j} f_j(k) - \partial_{x_i} f_i(k)) + \mu y_i(k),\\
y_i(k+1) &= x_i(k+1) - x_i(k),
\end{align}  \normalsize
where $\eta>0$ is the step-rate, $k$ is the time step, $\mc{N}_i$ denotes the set of neighbours of agent $i$, $W_{ij}$ is the consensus weighting factor on the shared information over link $(j,i)$,  $0\leq \mu<1$ is the momentum-rate, and variable $y_i(k)$ denotes the momentum term of node $i$ at time $k$. \textit{The momentum term $y_i$ is added to improve the convergence rate and accelerate the solution.} For initialization, the agents set their initial state values as $x_i(0)=b_i$ and $y_i(0)=0$.

\subsection{Solution Subject to Sector-Bound Nonlinearity}
The data transmission channels (or the communication links) between agents might be subject to nonlinear constraints. This implies that the sent data $\partial_{x_j} f_j(k)$ over a link $(j,i)$ is delivered to agent $i$ as $g_l(\partial_{x_j} f_j(k))$, where $g_l: \mathbb{R} \mapsto \mathbb{R}$ denotes the nonlinear mapping. Then, the decentralized solution is updated as,

\small \begin{align} \nonumber
x_i(k+1) =& x_i(k) + \mu y_i(k) \\&+ \eta  \sum_{j \in \mc{N}_i}   W_{ij}  (g_l(\partial_{x_j} f_j(k)) -g_l(\partial_{x_i} f_i(k))), \label{eq_nonlin1}\\
y_i(k+1) =& x_i(k+1) - x_i(k), \label{eq_nonlin2}
\end{align} \normalsize	
where the weight matrix $W$ is only weight-balanced and not necessarily stochastic.
\begin{ass}
The nonlinear function $g_l(\cdot)$ is assumed to be odd, strongly sign-preserving, and sector-bound.	
\end{ass}
One example of such link nonlinearity is saturation or clipping. Another example is log-scale quantization \cite{TASE_quant,LCSS_quant} illustrated in Fig.~\ref{fig_quant} with its formula  as follows:
\begin{align}
g_{l}(u) &:= \mbox{sgn}(u) \exp(g_{u}(\log(|u|))), \label{eq_quan_log}
\end{align}
with $g_{u}(u) := \rho \left[ \frac{u}{\rho}\right]$ (the operator $\left[\cdot\right]$ as rounding to the nearest integer) and $\rho>0$ as the quantization level.
\begin{figure}[]
\centering
\includegraphics[width=2.5in]{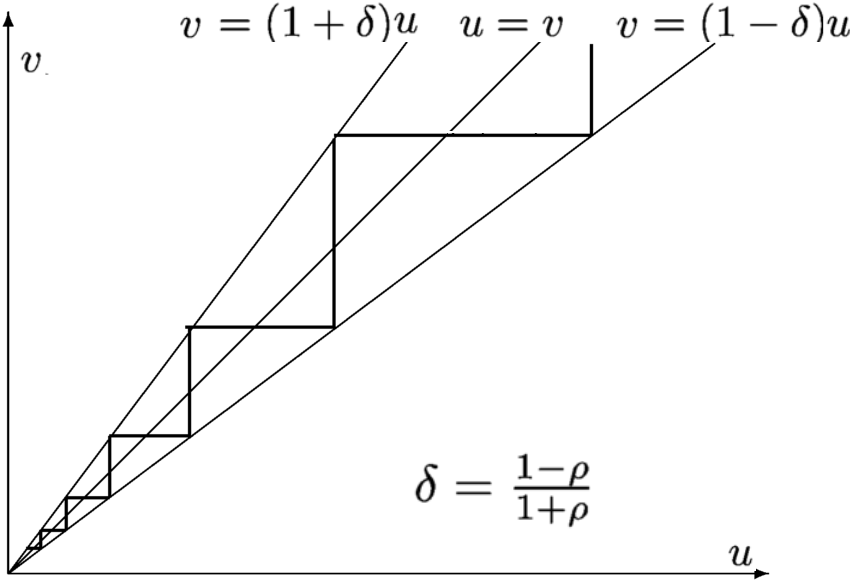}
\caption{The log-scale or logarithmic quantizer (with the quantization level $\rho$) is shown in this figure as a strongly sign-preserving sector-bound nonlinearity. The sector-bounds of this nonlinear function are $1 \pm \delta$ as shown in the figure.  }
\label{fig_quant}
\end{figure}

\subsection{Delay-Tolerant Solution}
Latency in a network refers to the delay between the transmission of data from a source agent to its destined agent; in other words, the exchanged data packets over the communication network are typically subject to time-delays. This might be due to limited bandwidth which leads to increased latency and drops, especially under high load. Further, mismanaged traffic and congestion, synchronization issues, and physical obstructions and interference from other devices can contribute to latency over the network. This may lead to instability in the multi-agent network and, in distributed optimization setup, is described by the number of time steps it takes to transfer data over a link/channel. Let $\overline{\tau}$ denote the maximum possible steps of time-delay over the network and  $\tau_{ij}(k) \leq \overline{\tau}$ denote the delay over the link $(j,i)$ at time $k$. We assume the most general case for the delays as arbitrary, random, (possibly) time-varying, and heterogeneous at different links.
Then, the delay-tolerant solution is updated as
\begin{align} \nonumber
x_i(k+1) =& x_i(k) +\eta_\tau  \sum_{j \in \mc{N}_i}  \sum_{r=0}^{\overline{\tau}} W_{ij}  \Big(g_l(\partial_{x_j} f_j(k-r)) \\ &-g_l(\partial_{x_i} f_i(k-r))\Big) \mathcal{I}_{k-r,ij}(r) + \mu y_i(k), \label{eq_sol_delay1}\\
y_i(k+1) =& x_i(k+1) - x_i(k), \label{eq_sol_delay2}
\end{align}
with $\eta_\tau$ as the step-rate, $\overline{\tau} \in \mathbb{Z}^+$ as the maximum possible time-delay (upper bound) and the indicator function $\mathcal{I}$ defined as,
\begin{align} \label{eq_I}
\mathcal{I}_{k,ij}(r) = \left\{ \begin{array}{ll}
	1, & \text{if}~  \tau_{ij}(k) = r,\\
	0, & \text{otherwise}.
\end{array}\right.
\end{align}
It should be noted that in Eq. \eqref{eq_sol_delay1}, although the local information is immediately available, the delayed information $g_l(\partial_{x_j} f_j(k-r))-g_l(\partial_{x_i}f_i(k-r))$ is used instead of $g_l(\partial_{x_j} f_j(k-r))-g_l(\partial_{x_i}f_i(k))$. In other words, agents use their past state information $g_l(\partial_{x_i}f_i(k-r))$ based on the time-delay $0\leq r \leq \overline{\tau}$ to ensure that the sum of states remains constant. This guarantees all time resource-demand feasibility as proved later in Lemma~\ref{lem_feas}, which ensures that at any termination time of the algorithm the balance between assigned resources $\sum_{i=1}^{n} x_i(k)$ and the demand $\sum_{i=1}^{n} b_i$ holds irrespective of the time-delays. For example, for application in distributed energy resource management and generator coordination, this implies that the allocated power states to the generators are always equal to the power demand and there would be no service disruption due to delays. 

The most general form of the proposed iterative solution is summarized in Algorithm~\ref{alg_1}. Note that in Eq. \eqref{eq_sol_delay1} the amount of time-delay of a received message is known to the destined agent. This is done by time-stamping the data shared over the network. 
\begin{algorithm}
\textbf{Input:}  $\mc{N}_i$, $W$, $\mc{G}$, $\eta_\tau$, $\mu$,  $b_i$, $f_i(\cdot)$, $\overline{\tau}$\;
\textbf{Initialization:} set $k=0$, $x_i(0)=b_i$,  $y_i(0)=0$\;
\While{termination criteria NOT true}{
	Agent $i$ receives (possibly delayed) information $\partial_{x_j} f_j$ from agent $j \in \mc{N}_i$ over the network $\mc{G}$\;
	Agent $i$ updates its resource state via Eq.~\eqref{eq_sol_delay1}-\eqref{eq_sol_delay2}\;
	Agent $i$ shares information $\partial_{x_i} f_i$ with neighbors $i \in \mc{N}_j$ over the network $\mc{G}$\;
	$k \leftarrow k+1$\;
}
\textbf{Return} Assigned resource state $x_i$ and optimal cost $f_i(x_i)$\;	
\caption{Distributed Multi-Agent Resource Scheduling}
\label{alg_1}
\end{algorithm}
\section{Analysis of Convergence} \label{sec_conv}
This section analyzes the convergence of the proposed algorithm towards optimal resource scheduling. We prove all-time constraint feasibility, stability of the optimizer $\mb{x}^*$ (as the equilibrium of the proposed dynamics), and convergence condition towards this optimal point.

\begin{lem} [Constraint Feasibility] \label{lem_feas}
Given a uniformly connected undirected network of agents, the solution under Algorithm~\ref{alg_1} preserves its all-time constraint feasibility.
\end{lem}
\begin{proof} We need to show that $\sum_{i=1}^{n} (x_i(k) -b_i)=0$ for all $k$. From the proposed dynamics \eqref{eq_sol_delay1}-\eqref{eq_sol_delay2} we have,

\small	\begin{align} \nonumber
	\sum_{i=1}^{n} x_i(k+1) =& \sum_{i=1}^{n} x_i(k) +\eta_\tau  \sum_{i=1}^{n} \sum_{j \in \mc{N}_i}  \sum_{r=0}^{\overline{\tau}} W_{ij}  \Big(g_l(\partial_{x_j} f_j(k-r)) \\ \nonumber &-g_l(\partial_{x_i} f_i(k-r))\Big) \mathcal{I}_{k-r,ij}(r) + \mu \sum_{i=1}^{n} y_i(k)
\end{align} \normalsize
Recall that the weight matrix is symmetric and the nonlinear mapping is odd and sign-preserving. Therefore,
\begin{align} \nonumber
	W_{ij} \Big(g_l(\partial_{x_i} &f_i(k-r)) -  g_l(\partial_{x_j} f_j(k-r))\Big) = \\ \nonumber &-W_{ji} \Big(g_l(\partial_{x_j} f_j(k-r) -  g_l(\partial_{x_i} f_i(k-r)\Big).
\end{align}
As a result, the summation on the gradient tracking is zero and $\sum_{i=1}^{n} x_i(k+1) = \sum_{i=1}^{n} x_i(k) + \mu \sum_{i=1}^{n} y_i(k)$. Initializing from $x_i(0)=b_i$ and $y_i(0)=0$ and substituting $y_i(k) = x_i(k)-x_i(k-1)$, we have $\sum_{i=1}^{n} x_i(k+1) = \sum_{i=1}^{n} x_i(k)$. Further, we have $\sum_{i=1}^{n} y_i(k) = 0$. This completes the proof.
\end{proof}

\begin{lem} [Equilibrium] \label{lem_tree}
Let $\mb{x}^*$ be the optimal point of the Algorithm~\ref{alg_1} and the equilibrium of the proposed dynamics. Then, following Lemma~\ref{lem_z*}, $\nabla F(\mb{x}^*) \in \mbox{span}\{\mb{1}_n\}$.
\end{lem}
\begin{proof} We prove this by reductio ad absurdum. Let
$\nabla F(\mb{x}^*) = ({\varphi}^*_1;\dots;{\varphi}^*_n)$ where $ \varphi^* \notin \mbox{span}\{\mb{1}_n\}$, i.e.,
\begin{align}
	\varphi^*_i \neq \varphi^*_j \Longleftrightarrow  \partial f_i(x_i^*)\neq \partial f_j(x_j^*),
\end{align}
for at least two nodes $i,j$. Let,
\begin{align}
	\alpha = \argmax_{q\in \{1,\dots,n\}}  {\varphi}^*_{q},~
	\beta = \argmin_{q \in \{1,\dots,n\}}  {\varphi}^*_q. \label{eq_beta0}
\end{align}
as two nodes over the network with maximum and minimum gradients. 	
Recall that $\mc{G}_B$ is uniformly connected; therefore,  there exists a path in $\mc{G}_B$ from node $\alpha$ to node $\beta$ over time-interval $B$. In this path, there is at least two nodes $\overline{\alpha}$ and $\overline{\beta}$  with the neighbors $\mc{N}_{\overline{\alpha}}$ and $\mc{N}_{\overline{\beta}}$, respectively, satisfying
$
\varphi^*_{\overline{\alpha}} > \varphi^*_{\mc{N}_{\overline{\alpha}}},~ \varphi^*_{\overline{\beta}} < \varphi^*_{\mc{N}_{\overline{\beta}}}
$.  This implies that the gradient-tracking summation in dynamics \eqref{eq_sol_delay1}  is non-zero, and  ${x}^*_{\overline{\alpha}}(k+1) \neq {x}^*_{\overline{\alpha}}(k)$ and ${x}^*_{\overline{\beta}}(k+1) \neq {x}^*_{\overline{\beta}}(k)$ over time-window of length $B$. This contradicts the definition of the equilibrium and completes the proof.
\end{proof}

\begin{theorem} [Convergence] \label{thm_delay}
The states of the agents under Algorithm~\ref{alg_1} collaboratively converge to the optimal solution of problem \eqref{eq_dra} for $0\leq\mu<1$ and
\begin{align} \label{eq_eta2}
	\eta_\tau <  \frac{\kappa \lambda_2}{u \lambda_n^2 \mc{K}^2(\overline{\tau}+1)}
\end{align}
\end{theorem}

\begin{proof} We prove the theorem in three steps.

\textbf{Step I:} we set $\mu=0$ and derive the convergence condition for the delay-free case ($\overline{\tau}=0$), i.e., solution under Eqs. \eqref{eq_nonlin1}-\eqref{eq_nonlin2}.
Let define $\overline{F}(k) := F(\mb{x}(k))-F(\mb{x}^*)$ with $F(\mb{x}^*)$ as the optimal cost. We need to prove that this residual satisfies $\overline{F}(k+1) < \overline{F}(k)$ for $\mb{x}(k)  \neq \mb{x}^*$. Recall that ${\mb{y}(k+1) = \mb{x}(k+1)-\mb{x}(k)}$.
Following Lemma~\ref{lem_strict}, we have
\begin{align} \label{eq_proof_1}
	\nabla_x F^\top \mb{y}(k+1)  + u \mb{y}(k+1)^\top \mb{y}(k+1)  \leq 0.
\end{align}
Using the graph-theoretic notions of Section~\ref{sec_graph} and notations from the consensus literature \cite{olfatisaberfaxmurray07}, we rewrite Eq.~\eqref{eq_sol_delay1} (with some abuse of notation) in the Laplacian-gradient form as
\begin{align}
	\mb{x}(k+1) = \mb{x}(k) - \eta_\tau  L g_l(\nabla_x F) \label{eq_L}
\end{align}
with $g_l(\nabla_x F)$ acting as an element-wise nonlinear mapping on entries of $\nabla_x F$. For notation simplicity, we drop the dependence on $k$.
Then, from Eq.  \eqref{eq_proof_1} and \eqref{eq_L}, we have
\begin{align} \label{eq_proof_2}
	-\eta_\tau  \nabla_x F^\top     L g_l(\nabla_x F)  + u  \eta_\tau^2     g_l(\nabla_x F)^\top L^\top L g_l(\nabla_x F) \leq 0.
\end{align}
Recalling the equilibrium condition from Lemma~\ref{lem_tree}, define $\phi(\mb{x}) := \nabla_x F -  \frac{\mb{1}_n}{n}\sum_{i=1}^n \partial_{x_i} f_i $ as the dispersion (or disagreement) parameter. Using Lemma~\ref{lem_xLy}, it is straightforward to show that Eq.~\eqref{eq_proof_2} holds for
\begin{align}
	\label{eq_proof__rho}
	(-\kappa \eta \lambda_2 + u \lambda_n^2 \mc{K}^2 \eta^2)    \phi^\top    \phi \leq 0,
\end{align}
Because $\phi^\top    \phi > 0$ for $\mb{x} \neq \mb{x}^*$,  we need the following to hold
\begin{align} \label{eq_eta}
	\eta <  \frac{\kappa \lambda_2}{u \lambda_n^2 \mc{K}^2} =: \overline{\eta}.
\end{align}
This proves that for $\eta < \overline{\eta}$  the optimal convergence holds and the summation $\sum_{j \in \mc{N}_i}   W_{ij}  (g_l(\partial_{x_j} f_j) -g_l(\partial_{x_i} f_i))$ goes to zero.

\textbf{Step II:} Next, we extend the proof for convergence in the presence of time-delays ($\overline{\tau}>0$), i.e., the solution under Eqs. \eqref{eq_sol_delay1}-\eqref{eq_sol_delay2}.
For general arbitrary delays, the information sent from any node $i$ reaches its neighbour $j\in \mc{N}_i$ at most in $\overline{\tau}+1$ steps. Then, we need to consider  $\mb{x}(k+\overline{\tau}+1)-\mb{x}(k)$ for the convergence analysis in \textbf{Step I} instead of $\mb{x}(k+1)-\mb{x}(k)$. Recall from~\eqref{eq_I} that
$0 <  \sum_{r=0}^{\overline{\tau}} \mc{I}_{k-r,ij}(r) \leq (\overline{\tau}+1)$. This implies that the Eqs.~\eqref{eq_proof_1} and~\eqref{eq_proof_2} need  to be scaled by $\overline{\tau}+1$, and therefore, the outcome of $\eta_\tau$ is down-scaled by $\overline{\tau}+1$, i.e., $  (\overline{\tau}+1) \eta_\tau< \overline{\eta}$.
Therefore, to reach the convergence, $\eta_\tau$ needs to satisfy
\begin{align} \nonumber
	\eta_\tau <   \frac{\overline{\eta}}{\overline{\tau}+1}
\end{align}
For this range of $\eta_\tau$, we have $\overline{F}(k+1) < \overline{F}(k)$ and the residual converges to zero in the presence of time-delays. Further, under this range of $\eta_\tau$ we have the following,
\begin{align} \nonumber \sum_{j \in \mc{N}_i}  \sum_{r=0}^{\overline{\tau}} &W_{ij}  (g_l(\partial_{x_j} f_j(k-r)) \\ &-g_l(\partial_{x_i} f_i(k-r))) \mathcal{I}_{k-r,ij}(r) \rightarrow 0. \label{eq_sum_zero}
\end{align}

\textbf{Step III:} as the final step, we consider $0<\mu<1$. Note that Eq.~\eqref{eq_sol_delay1} can be written as
\begin{align} \nonumber
	y_i(k+1) = &\eta_\tau  \sum_{j \in \mc{N}_i}  \sum_{r=0}^{\overline{\tau}} W_{ij}  \Big(g_l(\partial_{x_j} f_j(k-r)) \\ &-g_l(\partial_{x_i} f_i(k-r))\Big) \mathcal{I}_{k-r,ij}(r) + \mu y_i(k),
\end{align}
Following Eq.~\eqref{eq_sum_zero}, the above converges in limit to a geometric series in the form $y_i(k+1) =  \mu y_i(k)$. For $0<\mu<1$ this series converges to zero. Further, from the proof of Lemma~\ref{lem_feas} we have
\begin{align} \nonumber
	\eta_\tau  \sum_{i=1}^{n} \sum_{j \in \mc{N}_i}  \sum_{r=0}^{\overline{\tau}} &W_{ij}  \Big(g_l(\partial_{x_j} f_j(k-r)) \\ \nonumber &-g_l(\partial_{x_i} f_i(k-r))\Big) \mathcal{I}_{k-r,ij}(r) =0
\end{align}
Recalling the initialization $y_i(0)=0$, this further implies that
\begin{align} \nonumber
	\sum_{i=1}^{n} y_i(k+1) =  \sum_{i=1}^{n} \mu y_i(k)=0
\end{align}
for all $k$. This completes the proof.
\end{proof}

\subsection{Some Discussions and Remarks}
It is worth mentioning that one can prove convergence and constraint-feasibility of the proposed algorithm over \textit{WB directed networks}  in the \textit{delay-free} case (i.e., with $\overline{\tau}=0$). Recall that Lemma~\ref{lem_xLy} holds for strongly-connected WB directed networks by replacing $L$ with $L_s:=\frac{L+L^\top}{2}$. Similar stability results for consensus algorithms are proved in \cite{SensNets:Olfati04}. Then, following as in consensus literature \cite{olfatisaberfaxmurray07} and substituting $\lambda_{ns},\lambda_{2s}$ (as the largest and smallest non-zero eigenvalues of $L_s$ respectively) the proof of Lemma~\ref{lem_xLy} can be extended to WB networks. Consequently, the proofs in Theorem~\ref{thm_delay} similarly follow for WB directed networks for $\overline{\tau}=0$.

Moreover, it should be mentioned that the proposed nonlinear solution not only works for log-scale quantization and saturation but also works for any sector-bound sign-preserving nonlinearity. This implies that one can purposefully apply various proper nonlinear mappings, for example, to tune the convergence rate via sign-based functions. Further, similar to nonlinear robust consensus algorithms \cite{shi2020benefits,wang2022robust}, other types of nonlinearity might be applied to dynamics \eqref{eq_nonlin1} to come up with robust solutions to disturbance and noise.

In contrast to the primal-dual formulations (such as ADMM) which require all-time connected static networks \cite{carnevale2022gtadam,jian2019distributed,banjac2019decentralized,jiang2022distributed}, this work allows for uniform-connectivity, i.e., changes in the network topology and temporary disconnectivity of the communication network.   Further, convergence over networks subject to time-delays and nonlinear data transmissions are hard to address via primal-dual formulations. Few existing delay-tolerant papers in the literature \cite{wang2022distributed2,liu2022distributed,cdc_dtac} are not all-time feasible and lose resource-demand balance at some times.

\section{Simulations} \label{sec_sim}
\subsection{Academic Example: Comparison with the Literature}	
For MATLAB simulation, we consider a random network of ER graph (with linking probability of $p=25\%$) representing the multi-agent network of $n=20$ nodes. The agents locally optimize the following cost function:
\begin{align} \label{eq_f_quad}
f_i(x_i)=g_i x_i^2 + d_i x_i + a_i +\sigma [{x}_i - M_i]^+ + \sigma [m_i - {x}_i]^+
\end{align}
with random cost parameters $g_i \in (0,0.3]$, $d_i \in (0,10]$, $a_i \in (0,10]$, and box constraint parameters $m_i=10$, $M_i=110$, $\sigma = 1$. For the balancing constraint, we consider $b_i = 50$ for all $i$. The step rate and momentum-rate are set as $\eta = 0.04$ and $\mu = 0.9$.

First, we simulate in the absence of time-delays (i.e., Eqs. \eqref{eq_nonlin1}-\eqref{eq_nonlin2}) and subject to log-scale quantization with quantizer level $\rho =\frac{1}{2^{10}}$. Under our proposed momentum-based nonlinear solution, the residual is compared with linear \cite{boyd2006optimal}, accelerated linear \cite{shames2011accelerated} (with $\beta=0.5$), single-bit \cite{taes}, finite sign-based \cite{chen2016distributed} (with $\zeta = 0.6$), fixed sign-based \cite{liu2023distributed} (with $\zeta_1 = 0.6,\zeta_2 = 1.5$), and saturated dynamics \cite{scl} (with $\delta=0.5$) in Fig.~\ref{fig_compare}. These works provide constraint-feasible solutions while many other existing literature, including ADMM-based formulations, are not all-time feasible.
\begin{figure}[]
\centering
\includegraphics[width=2.75in]{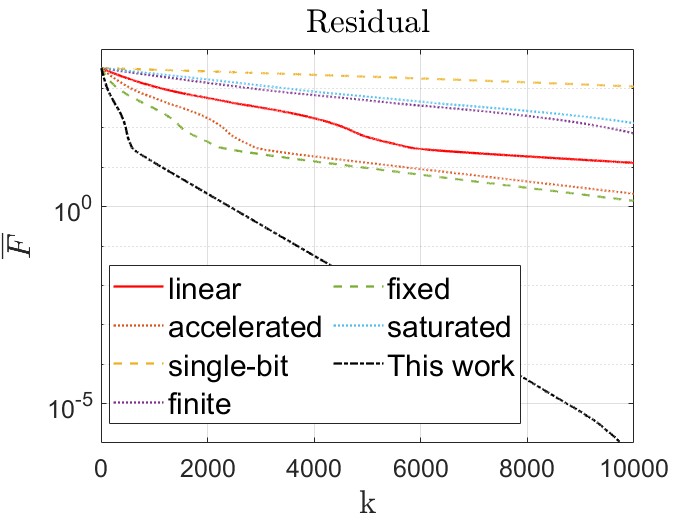}
\caption{ Comparison of the time-evolution of the cost residuals under different resource scheduling solutions in the literature.
}
\label{fig_compare}
\end{figure}

Next, we perform the simulation of the proposed nonlinear dynamics considering logarithmic quantization for different quantizer levels $\rho$. We set $\mu=0.5$, $\eta = 0.04$ and the rest of the simulation setup and parameters are the same as the previous ones. Fig.~\ref{fig_simquant} shows the residual convergence for different values $\rho$.
\begin{figure}[]
\centering
\includegraphics[width=2.75in]{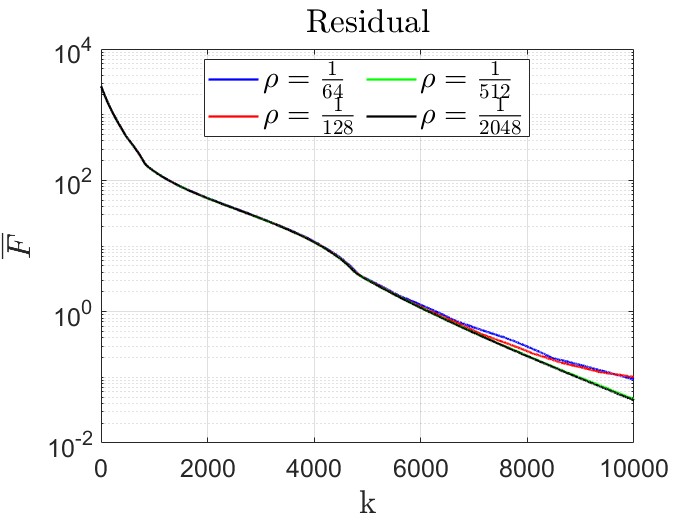}
\caption{ The time-evolution of the cost residuals under different log-scale quantization levels as sector-bound nonlinearity.
}
\label{fig_simquant}
\end{figure}

Next, we simulate our proposed algorithm over a uniformly-connected ER network setup by considering \textit{random} link failure rate $p_l = 80\%$. For linking probability $p=25\%$ of ER graph, the bond-percolation threshold for link failure from Eq.~\eqref{eq_bond} is $p_c = 58\%$. For link failure $p_l>p_c$ the ER graph becomes disconnected, but it preserves uniform-connectivity over time-window of $B$ steps for $p_l^B<p_c$. In other words, the \textit{union graph} is connected over $B\geq3$ steps. The simulation is performed for different momentum rates to show how it affects the convergence rate of the algorithm. The residuals along with momentum variables and resource states at all agents (and their average for constraint-feasibility check) are shown in Fig.~\ref{fig_momentum} versus time $k$. For this simulation, we set $\eta = 0.1$.
\begin{figure*}[]
\centering
\includegraphics[width=1.75in]{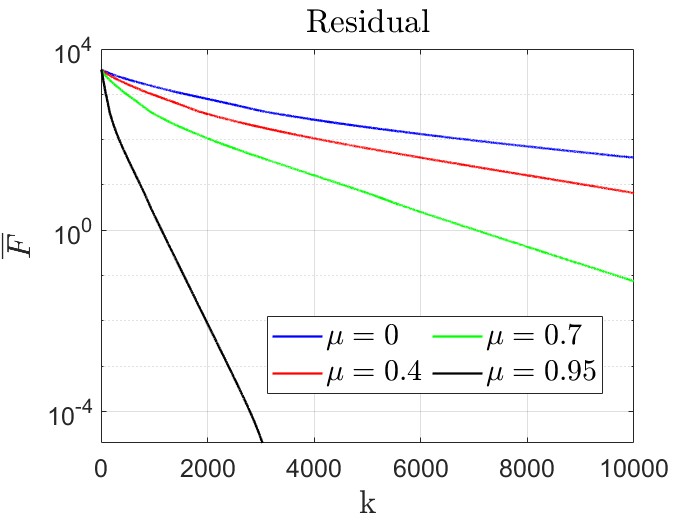}
\includegraphics[width=1.75in]{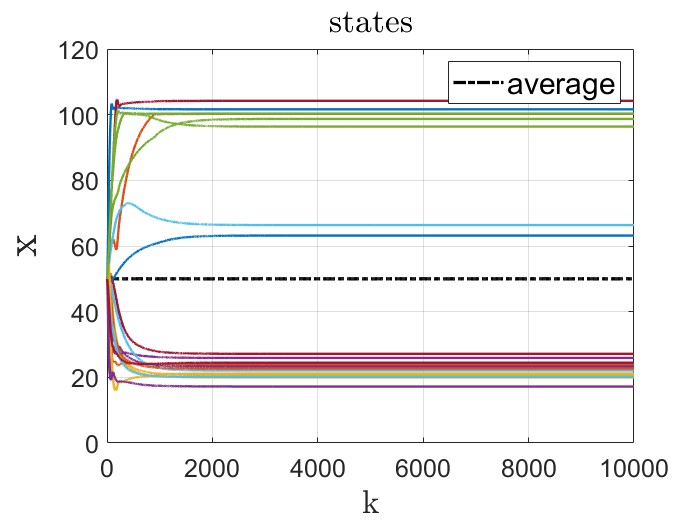}
\includegraphics[width=1.75in]{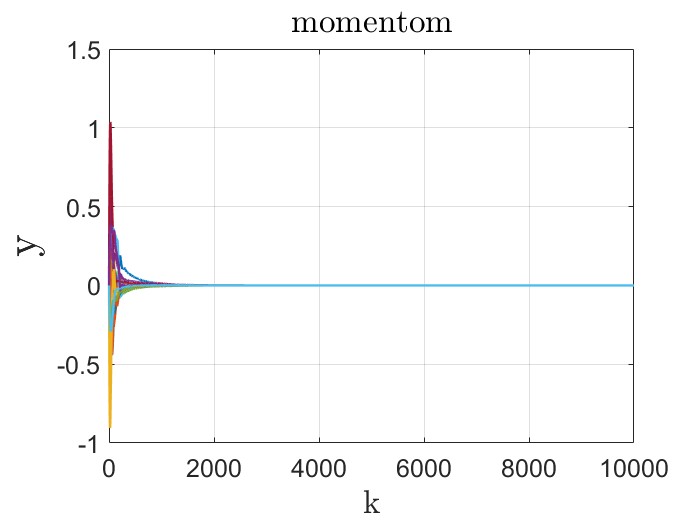}
\caption{(Left) The time-evolution of the scheduling cost residuals under different momentum rates. (Middle) The time-evolution of the resource states $x_i$ at all agents for $\mu=0.95$. The average of all states remains constant over time implying constraint-feasibility. (Right)  The time-evolution of the momentum states $y_i$ at all agents for $\mu=0.95$.
}
\label{fig_momentum}
\end{figure*}

For the last simulation, we illustrate the convergence subject to data-transmission time-delays in Fig.~\ref{fig_delay}. We used the following  parameters of the proposed algorithm: $\eta = 0.2$, $\mu = 0.8$, $p=20\%$, and different $\overline{\tau}$ values (see the figure). For the step rates and time-delay bounds not satisfying Eq.~\eqref{eq_eta2} there is oscillating steady-state residual cost and the solution is not convergent. It should be noted that the assigned resources preserve all-time constraint-feasibility (resource-demand balance) even subject to time-delays over the network. This advances the existing literature on distributed optimization in the presence of delays. As proved in Lemma~\ref{lem_feas} the solution satisfies constraint-feasibility at all times. Also, according to \textbf{Step III} of Theorem~\ref{thm_delay} all momentum states converge to zero.

\begin{figure*}[]
\centering
\includegraphics[width=1.75in]{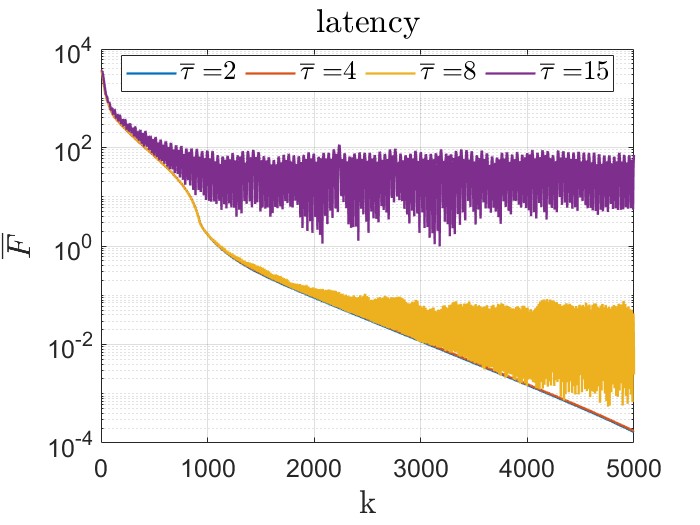}
\includegraphics[width=1.75in]{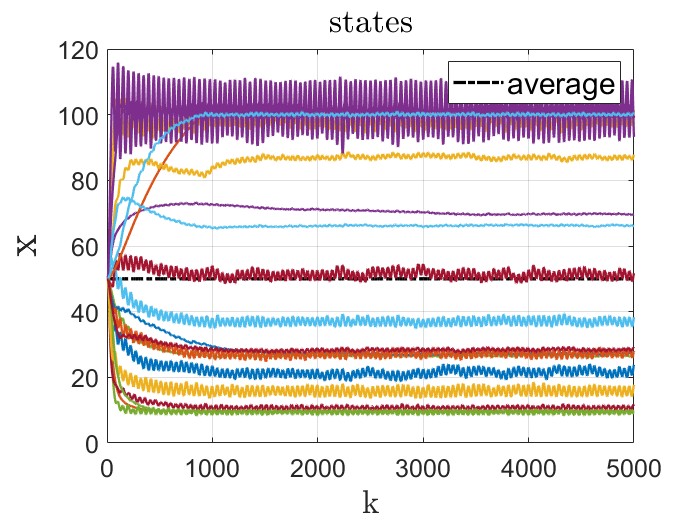}
\includegraphics[width=1.75in]{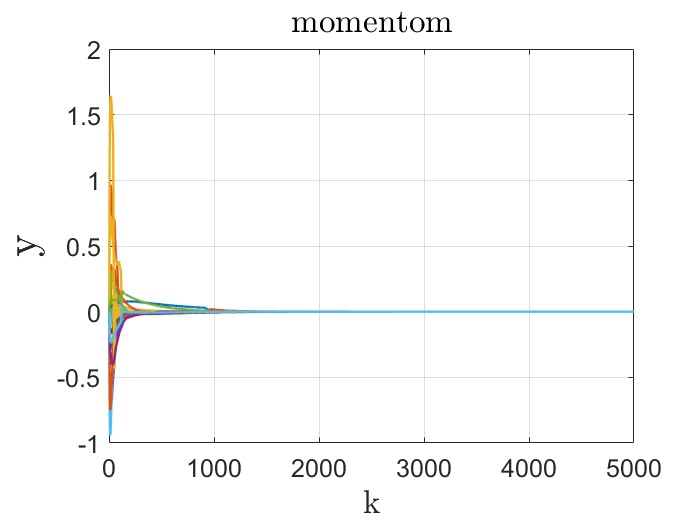}
\caption{(Left) The scheduling cost residuals versus time subject to different maximum time-delay values. (Middle) The time-evolution of the assigned resources (states) $x_i$ at all agents for $\overline{\tau}=4$. The average of all resource states remains constant even in the presence of time-delays, which implies all-time constraint feasibility. (Right)  The time-evolution of the momentum states $y_i$ at all agents subject to time-delays with $\overline{\tau}=4$. All momentums converge to zero as proved in \textbf{Step III} of Theorem~\ref{thm_delay}.
}
\label{fig_delay}
\end{figure*}

\subsection{Application to CPU Scheduling}	
In this section, we consider a real-world application in CPU scheduling over distributed data centres (servers). For this example, we consider a network of $n=100$ servers that need to be optimally scheduled with CPU workloads to optimize the following objective function \cite{rikos2021optimal,grammenos2023cpu}:
\begin{align}\label{eq:fiz2}
f_i(x_i) = \frac{1}{2\pi_i^{\max}} (x_i - \rho_i)^2
\end{align}
where $x_i$ denotes the assigned CPU workload and $\rho_i$ is the positive demand at node (server) $i$. As the coupling constraint, we have $\sum_{i=1}^n x_i  = \rho = \sum_{i=1}^n \rho_i$, implying that the sum of all workloads is equal to the overall CPU demand denoted by $\rho$. In the formulation~\eqref{eq:fiz2}, we have $\pi_i^{\max} \coloneqq c_i T_h$, with $c_i$ as the sum of all clock rate frequencies of all processing cores of server $i$ (in cycles-per-second) and $T_{h}$ is the CPU allocation time period. Here, we set $\pi_i^{\max}=100$, $\rho_i \in [15~35]$, and $\rho = 2500$. The local constraints are $ m_i = 0\leq x_i \leq  0.6 \pi_i^{\max} = M_i$ which limits the allocated CPU workload at each server $i$. To address these box constraints logarithmic penalty terms in the form~\eqref{eq_log_penalty} with $\alpha=2$ and $\sigma=4$ are considered. The network of servers is considered an ER random network with $p=12\%$ linking probability. Other parameters of the proposed solution are considered as $\eta = 0.1$, $\mu = 0.4$.  We compared our solution under logarithmically quantized channels (as the sector-bound nonlinearity) with uniformly quantized CPU scheduling in \cite{rikos2021optimal}. The quantization level is considered equal to $\frac{1}{2^4}$. The results are shown in Fig.~\ref{fig_cpu}. As shown in the figure, our proposed log-quantized solution converges to lower residuals over time as compared to the uniformly-quantized solution in \cite{rikos2021optimal}. In addition, the average of the CPU workloads is constant over time which implies that the balance between CPU workloads and demands holds at all times (this shows all-time constraint feasibility).

\begin{figure}[]
\centering
\includegraphics[width=2.35in]{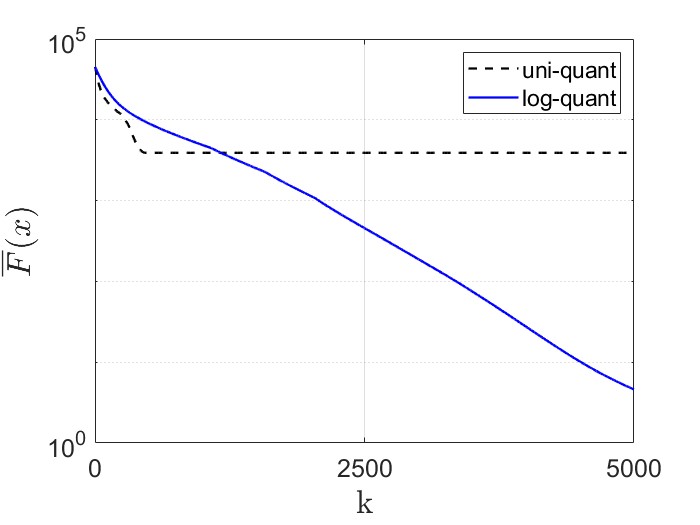}
\includegraphics[width=2.35in]{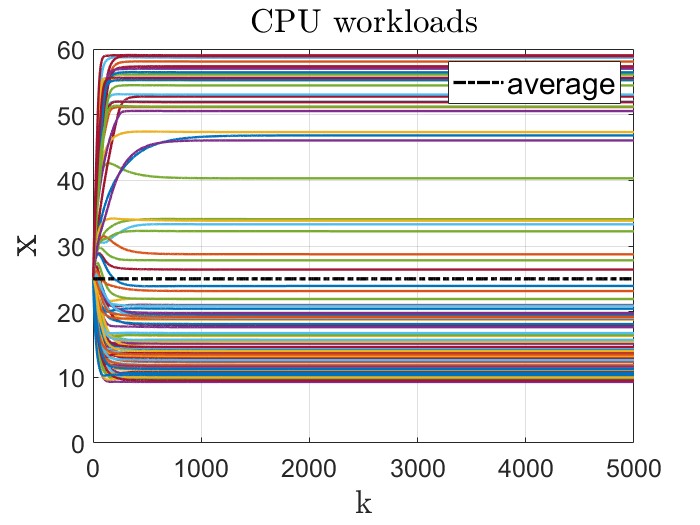}
\caption{ The left figure shows the residual of the CPU scheduling objective function under log-scale quantization (this work) and uniform quantization (the work \cite{rikos2021optimal}) over time. Clearly, uniform quantization may result in some steady-state residual in contrast to our work. The right figure shows the CPU workloads of $100$ servers over time. The average CPU workload is time-invariant, implying all-time CPU-demand feasibility.
}
\label{fig_cpu}
\end{figure}

\section{Conclusions} \label{sec_con}
\subsection{Concluding Remarks}
In this paper, we have presented a scalable and decentralized momentum-based dynamics to improve the convergence rate in distributed resource allocation scenarios. The paper addresses crucial challenges, including all-time resource-demand feasibility, communication time-delay, uniform network connectivity, and the potential nonlinearity induced by factors such as logarithmic quantization. These hold particular importance in the context of real-world systems characterized by fluctuating demands, diverse communication conditions, and limitations induced by nonlinear effects (e.g., quantization and clipping). In this direction, our proposed framework provides a versatile solution for resource scheduling in the context of distributed computing, not only advancing the theoretical understanding but also providing practical solutions to enhance the convergence rate.

\subsection{Future Directions}
The insights gained from this research pave the way for future innovations and optimizations, e.g., in machine learning and artificial intelligence; for example, one can address adaptive parameter tuning and fast convergence in the presence of nonlinearities induced by logarithmic quantization and saturation/clipping. Delay-tolerant machine learning scenarios are another interesting avenue of future research. One can further apply the proposed scenarios to the applications mentioned in Section~\ref{sec_app}, while addressing their particular problem-specific limitations; for example, the concept of ramp-rate-limits can be addressed in distributed energy resource management and economic dispatch via constraining the state rates by saturation function.

\section*{Acknowledgements}
This work has been supported by the Center for International Scientific Studies \& Collaborations (CISSC), Ministry of Science, Research and Technology of Iran.

\bibliographystyle{elsarticle-num}
\bibliography{bibliography}

\end{document}